\documentclass[12pt]{iopart}
\usepackage[all]{xy}
\usepackage{amsthm,amssymb,nath}
\usepackage{graphicx}
\def\Includegraphics[#1]#2{{\ptt[#1]#2}}

\let\epsilon\varepsilon
\def\hm{\hphantom{-}}
\def\eqref#1{{\rm(\ref{#1})}}
\newtheorem{proposition}{Proposition}

\newtheorem{corollary}{Corollary}

\newtheorem{remark}{Remark}
\newtheorem{example}{Example}
\newtheorem{construction}{Construction}

\def\dif{\mathop{}\!\mathrm d}

\def\sl{{\mathfrak{sl}}}

\begin{document}

\title{Nonlocal conservation laws of the constant astigmatism equation}

\author{Adam Hlav\'a\v{c} and Michal Marvan}
\address{Mathematical Institute in Opava, Silesian University in
  Opava, Na Rybn\'\i\v{c}ku 1, 746 01 Opava, Czech Republic.
  {\it E-mail}: Adam.Hlavac@math.slu.cz, Michal.Marvan@math.slu.cz}
\date{}
\ams{37K05, 37K25, 37K35}
\pacno{02.30.Ik}

\begin{abstract}
For the constant astigmatism equation, we construct a system of nonlocal conservation laws 
(an abelian covering) closed under the reciprocal transformations. 
We give functionally independent potentials modulo a Wronskian type relation.
\end{abstract}

\section{Introduction}

The constant astigmatism equation~\cite{B-M I} 
$$
\numbered\label{CAE}
z_{yy} + (\frac1z)_{xx} + 2 = 0
$$
represents surfaces characterised by constant difference between the 
principal radii of curvature, with $x,y$ being the curvature coordinates. 
The same equation represents orthogonal equiareal patterns on the unit sphere, closely 
related to sphere's plastic deformations within itself, see~\cite{H-M II}. 
Both topics are classical, see, e.g., Bianchi~\cite[\S~375]{B II}, 
although equation~\eqref{CAE} itself did not appear at that time. 

It is clear from the geometry that the constant astigmatism equation is transformable
to the sine--Gordon equation $\phi_{XY} = \sin\phi$~\cite{H-M II}.
Then, of course, the constant astigmatism equation itself
is integrable in the sense of soliton theory. Known are the zero curvature 
representation, see~\cite{B-M I} or eq.~\eqref{zcr} below, the bi-Hamiltonian 
structure and hierarchies of higher order symmetries and conservation laws~\cite{P-Z}, 
as well as multi-soliton solutions~\cite{Hl}.

According to our previous paper~\cite{H-M III}, the constant astigmatism equation has 
six local conservation laws with associated potentials $\chi,\xi,\eta,\zeta,\alpha,\beta$ 
satisfying
$$
\begin{gathered}
\chi_x = z_y + y, 
 & \chi_y = \frac{z_x}{z^2} - x, \\
\eta_x = x z_y, 
 & \eta_y = x \frac{z_x}{z^2} + \frac1z - x^2, \\
\xi_x = -y z_y + z - y^2, 
 & \xi_y = -y \frac{z_x}{z^2}, \\
\zeta_x = x y z_y - x z + \frac12 x y^2, 
 & \zeta_y = x y \frac{z_x}{z^2} + \frac y z - \frac12 x^2 y, \\
\alpha_x = \frac{\sqrt{(z_x + z z_y)^2 + 4 z^3}}{z}, 
 & \alpha_y = \frac{\sqrt{(z_x + z z_y)^2 + 4 z^3}}{z^2}, \\
\beta_x = \frac{\sqrt{(z_x - z z_y)^2 + 4 z^3}}{z}, 
 & \beta_y = -\frac{\sqrt{(z_x - z z_y)^2 + 4 z^3}}{z^2}.
\end{gathered}
$$
Potentials $\alpha,\beta$ correspond to the independent variables $X,Y$ of the related 
sine--Gordon equation~\cite{B-M I}. 

Potentials $\xi,\eta$ are images of $x,y$ under the reciprocal transformations~\cite{H-M III} 
$\mathcal Y, \mathcal X$, respectively; see formulas~\eqref{RT} below. 
Applying $\mathcal X$ to $\xi$ and $\mathcal Y$ to $\eta$, we obtain new nonlocal 
potentials and the process can be continued indefinitely.
It is then natural to ask what is the minimal set of potentials closed under the action of 
$\mathcal X$ and $\mathcal Y$. 
They are also nonlocal conservation laws of the sine--Gordon equation, but
available descriptions~\cite{S-B,W-T-L-L} are not of much help.

In this paper we first generate hierarchies of nonlocal conservation laws from the zero 
curvature representation and then look how they are acted upon by $\mathcal X$ and 
$\mathcal Y$. This allows us to find a set of potentials closed under the 
reciprocal transformations rather easily.

It has been known from the very beginning of the soliton theory that hierarchies of 
conservation laws arise through expansion in terms of the spectral parameter~\cite{M-G-K}. 
Literature on the subject is vast any many ways to connect integrability and
hierarchies of conservation laws have been proposed
(see, for example,~\cite{Sas,C-T} or~\cite[Prop. 1.5]{D-S} or~\cite[Sect.~5d]{New}). 

However, the constant astigmatism equation belongs to the rare cases when nonlocal
conservation laws can be obtained in a very straightforward way, almost effortlessly.
It is also easily seen how they are acted upon by the transformations.

\section{Preliminaries}


Let $\mathcal E$ be a system of partial differential equations in two independent variables $x,y$.
A conservation law is a 1-form $f \dif x + g \dif y$ such that $f_y - g_x = 0$ as a consequence 
of the system~$\mathcal E$. 
A {\it potential}, say $\phi$, corresponding to this conservation law 
is a variable which formally satisfies the compatible system $\phi_x = f$, $\phi_y = g$.

Let $\mathfrak g$ be a matrix Lie algebra.
A $\mathfrak g$-valued {\it zero curvature representation}~\cite{Z-S} of the system $\mathcal E$ 
is a 1-parametric family of $\mathfrak g$-valued forms 
$\alpha(\lambda) = A(\lambda) \dif x + B(\lambda) \dif y$ such that 
$A_y - B_x + [A,B] = 0$ as a consequence of the system~$\mathcal E$. 

Let $Q$ be an arbitrary matrix (called a {\it gauge matrix}) belonging to the associated Lie 
group~$\mathcal G$.
The {\it gauge transformation}~\cite{Z-S} with respect to $Q$ sends  
$\alpha = A \dif x + B \dif y$ to
${}^Q\alpha = {}^Q\!A \dif x + {}^Q\!B \dif y$, where
$$
\numbered\label{gauge}
{}^Q\!A = Q_x Q^{-1} + Q A Q^{-1}, \quad
{}^Q\!B = Q_y Q^{-1} + Q B Q^{-1}.
$$
We also say that ${}^Q\!A \dif x + {}^Q\!B \dif y$ is {\it gauge equivalent} to 
$A \dif x + B \dif y$.

The shortest way to conservation laws is from a 
zero curvature representation that vanishes at some value $\lambda_0$ of $\lambda$.
Without loss of generality we assume that $\lambda_0 = 0$, i.e., $A(0) = B(0) = 0$.
Consider the associated compatible linear system~\cite{Z-S} 
(or a differential covering~\cite{K-V,B-V-V})
$$
\numbered\label{Phi 1 sys}
\Phi_x = A \Phi, \quad \Phi_y = B \Phi,
$$
where $\Phi$ is a column vector. 
Expanding $\Phi$ into the formal power series 
$$
\Phi = \sum_{i = 0}^\infty \Phi_{i} \lambda^i
$$
around zero and inserting into~\eqref{Phi 1 sys}, we obtain compatible equations
$$
\numbered\label{Phi ser sys}
\Phi_{n,x} = \sum_{i = 1}^n A_i \Phi_{n - i}, 
\quad 
\Phi_{n,y} = \sum_{i = 1}^n B_i \Phi_{n - i}, 
\quad n \ge 0,
$$
where $A_i$, $B_i$ are the coefficients of the Taylor expansion of $A, B$ 
around $\lambda = 0$.
Here we start from $i = 1$ since $A_0 = B_0 = 0$.
By formulas~\eqref{Phi ser sys}, each of the
derivatives $\Phi_{n,x}, \Phi_{n,y}$ is explicitly expressed in terms of 
$\Phi_{0}, \dots, \Phi_{n-1}$.
Moreover, $\Phi_{0,x} = \Phi_{0,y} = 0$, meaning that $\Phi_{0}$ is a constant vector.  
Choosing $\Phi_{0}$ suitably, we thus obtain what may be called a hierarchy 
of vectorial potentials
$$
\begin{gathered}
\Phi_{1,x} = A_1 \Phi_0, & \Phi_{1,y} = B_1 \Phi_0, \\
\Phi_{2,x} = A_1 \Phi_1 + A_2 \Phi_0, & \Phi_{2,y} = B_1 \Phi_1 + B_2 \Phi_0, \\
\Phi_{3,x} = A_1 \Phi_2 + A_2 \Phi_1 + A_3 \Phi_0, &
\Phi_{3,y} = B_1 \Phi_2 + B_2 \Phi_1 + B_3 \Phi_0, \\
\qquad\vdots
\end{gathered}
$$ 
The 1-forms
$$
\sum_{i = 1}^n A_i \Phi_{n - i} \dif x + \sum_{i = 1}^n B_i \Phi_{n - i} \dif y
$$
then constitute a hierarchy of vectorial conservation laws, linear in the 
potentials $\Phi_{i}$. Their components are the scalar conservation laws sought.
They are termed `nonlocal' since they depend on the potentials. 
The whole hierarchy of potentials is also a special abelian covering~\cite{B-V-V}.

\begin{example} \rm
The integrable Harry Dym equation $u_t = u^3 u_{xxx}$ 
(see, e.g.,~\cite{A-F,Dmi,H-B-C,H-S,Mag,Pav,W-I-S}) has the
zero curvature representation
$$
A = \left(\begin{array}{cc} 0 & 2 \lambda \\ -\frac{1}{2} \lambda/u^2 & 0 \end{array}\right),
\quad
B = \left(\begin{array}{cc} 2 \lambda^2 u_x & -8 \lambda^3 u \\ \lambda  u_{xx} + 2 \lambda^3/u & -2 
\lambda^2 u_x \end{array}\right)
$$
polynomial in $\lambda$ and vanishing at $\lambda = 0$. 
Using equations~\eqref{Phi ser sys}, we find an explicit infinite 
hierarchy of nonlocal conservation laws 
$$
\numbered\label{HD cov}
\begin{gathered}
q_{i,x}  = -\frac12 p_{i-1}/u^2, &
q_{i,t} = -u_x q_{i-2} + u_{xx} p_{i-1} + p_{i-3}/u, & i \ge 1,
\\
p_{i,x} = q_{i-1}, &
p_{i,t} = p_{i-2} u_x - 2 q_{i-3} u, & i \ge 2,
\end{gathered}
$$
undoubtedly known to experts in the field. 
Here $p_i = q_i = 0$ for $i < 0$ and $p_0 = q_0 = 1$, $p_1 = x$.
\end{example}

It can also happen that $A_0, B_0$ vanish modulo a gauge transformation. 
The condition can be easily revealed (see below), since zero curvature representations gauge 
equivalent to zero have zero characteristic element~\cite{zcr}. 

If $A_0, B_0$, instead of vanishing, belong to a solvable subalgebra of $\mathfrak g$, 
then the procedure remains essentially the same; $\Phi_{0}$ itself corresponds to a finite 
hierarchy of potentials~\cite{rzcr} (for instance, this is so for the famous mKdV equation).
Cases when a gauge equivalent zero curvature representation fits into a solvable subalgebra
are easy to recognise, too.
Such a zero curvature representation admits a local solution of the associated 
Riccati equation~\cite{rzcr}.

As is well known, linearly independent conservation laws can have functionally dependent
potentials; cf. the discussion of local or potential dependence in~\cite{P-I}. 

\begin{example} \rm
Continuing the Harry Dym example, the relations
$$
\sum_{i = 0}^{2n} (-1)^i q_i p_{2n - i} = c_n, \quad n \ge 0,
$$ 
among potentials are compatible with the system~\eqref{HD cov}. 
Here $c_n$ denote arbitrary constants. 
\end{example}

Below we shall be interested in how conservation laws change under various local and 
nonlocal symmetries of the equation. 
It is typical for an integrable equation that its symmetries either preserve the zero 
curvature representation or send it to a gauge-equivalent one (see~\cite{F-S} for an
infinitesimal criterion). If the zero curvature representation is preserved, then so are the hierarchies.
If a gauge transformation~\eqref{gauge} corresponds to the symmetry, then $\Phi$ 
transforms to $Q\Phi$, since~\eqref{Phi 1 sys} is equivalent to
$$
\numbered\label{QPhi 1 sys}
(Q\Phi)_x = {}^Q\!A Q \Phi, \quad (Q\Phi)_y = {}^Q\!B Q \Phi.
$$

\section{The zero curvature representation}
\label{sect:CAE}


From now on, we deal with the constant astigmatism equation~\eqref{CAE}.
The zero curvature representation 
we shall start with is $\alpha = A^\circ \dif x + B^\circ \dif y$, where
$$ 
\numbered\label{zcr}
A^\circ(\lambda) = \left(\begin{array}{cc} \frac{(1 + \lambda^2) z_x}{8 \lambda z}
 + \frac{(1 - \lambda^2) z_y}{8 \lambda} & 
\frac{(\lambda + 1)^2 \sqrt{z}}{4 \lambda} \\ \frac{(1 - \lambda)^2 \sqrt{z}}{4 \lambda} & 
-\frac{(1 + \lambda^2) z_x}{8 \lambda z} - \frac{(1 - \lambda^2) z_y}{8 \lambda} \end{array}\right),
\\
B^\circ(\lambda) = \left(\begin{array}{cc} 
\frac{(1 - \lambda^2) z_x}{8 \lambda z^2} + \frac{(1 + \lambda^2) z_y}{8 \lambda z} & 
\frac{1 - \lambda^2}{4 \lambda \sqrt{z}} \\ \frac{1 - \lambda^2}{4 \lambda \sqrt{z}} & 
-\frac{(1 - \lambda^2) z_x}{8 \lambda z^2} - \frac{(1 + \lambda^2) z_y}{8 \lambda z} \end{array}\right)
$$
are two $\sl(2)$ matrices rationally depending on~$\lambda$. 
The characteristic matrix $C^\circ$~\cite{dp}, as determined from the condition 
$$
A^\circ_y - B^\circ_x + [A^\circ,B^\circ] = [z_{yy}  + (1/z)_{xx} + 2]\cdot C^\circ,
$$
is
$$ 
C^\circ(\lambda) = \left(\begin{array}{cc} 
\frac{1 - \lambda^2}{8 \lambda} & 0 \\ 
0 & -\frac{1 - \lambda^2}{8 \lambda} 
\end{array}\right).
$$
We see that $C^\circ(\lambda)$ vanishes when $\lambda = \pm 1$. Consequently, both $\alpha(1)$ 
and $\alpha(-1)$ are gauge equivalent to zero, which makes them suitable for immediate generation 
of conservation laws. 
However, $A(\lambda) \dif x + B(\lambda) \dif y$ and $A(-\lambda) \dif x + B(-\lambda) \dif y$
turn out to be gauge equivalent and, therefore, $\lambda = \pm 1$ lead to equivalent results.
Choosing $\lambda = -1$, 
the trivialising gauge matrix is
$$ 
Q = \left(\begin{array}{cc} z^{1/4} & 0 \\ x z^{1/4} & z^{-1/4} \end{array}\right),
$$
since we have 
$$ 
A^\circ(-1) = \left(\begin{array}{cc} -\frac{z_x}{4 z} & 0 \\ -\sqrt{z} & \frac{z_x}{4 z}
\end{array}\right) = -Q^{-1} Q_x,
\quad
B^\circ(-1) = \left(\begin{array}{cc} -\frac{z_y}{4 z} & 0 \\ 0 & \frac{z_y}{4 z}
\end{array}\right) = -Q^{-1} Q_y.
$$
Otherwise said, matrices
$$
{}^Q\!A^\circ(\lambda) = Q_x Q^{-1} + Q A^\circ(\lambda) Q^{-1}, \quad
{}^Q\!B^\circ(\lambda) = Q_y Q^{-1} + Q B^\circ(\lambda) Q^{-1}
$$
vanish at $\lambda = -1$. 
Then
$$
A(\lambda) = {}^Q\!A^\circ(\lambda - 1), \quad 
B(\lambda) = {}^Q\!B^\circ(\lambda - 1)
$$
vanish at $\lambda = 0$, i.e., fit the assumption $A(0) = B(0) = 0$ of the previous 
section.
Explicitly,
$$ 
\numbered\label{AB}
A(\lambda) = \left(\begin{array}{cc} 
\frac{\lambda (\lambda - 2) K_1}{2 (\lambda - 1)} - \frac{\lambda^2 z L_1}{2 (\lambda - 1)} & 
\frac{\lambda^2 z}{4(\lambda - 1)}  
\\ 
\frac{\lambda (\lambda - 2) K_2}{2 (\lambda - 1)} - \frac{\lambda^2 z L_2}{2 (\lambda - 1)} & 
-\frac{\lambda (\lambda - 2) K_1}{2 (\lambda - 1)} + \frac{\lambda^2 z L_1}{2 (\lambda - 1)} 
\end{array}\right), 
\\
B(\lambda) = \left(\begin{array}{cc} 
\frac{\lambda (\lambda - 2)L_1}{2 (\lambda - 1)} - \frac{\lambda^2 K_1}{2 (\lambda - 1) z} & 
-\frac{\lambda (\lambda - 2)}{4 (\lambda - 1)} 
\\
\frac{\lambda (\lambda - 2) L_2}{2 (\lambda - 1)} - \frac{\lambda^2 K_2}{2 (\lambda - 1) z} & 
-\frac{\lambda (\lambda - 2)L_1}{2 (\lambda - 1)} + \frac{\lambda^2 K_1}{2 (\lambda - 1) z} 
\end{array}\right),
$$
where
$$ 
\numbered\label{KL}
\begin{gathered}
K_1 = -\frac{z_y}4, &
L_1 = -\frac{z_x}{4 z^2} + \frac{x}{2}, \\
K_2 = -\frac{x z_y} 2, &
L_2 = -\frac{x z_x}{2 z^2} - \frac{1}{2 z} + \frac{x^2}{2} .
\end{gathered}
$$
Thus, we can derive a double hierarchy of nonlocal conservation laws by
expansion of a 2-component vector $\Phi$ satisfying system~\eqref{Phi 1 sys}, i.e.,
$$
\numbered\label{Phi sys}
\Phi_x = A \Phi, \quad \Phi_y = B \Phi.
$$
This will be done in the next section.
Our primary interest in this section are the transformation properties of these hierarchies
under local and nonlocal symmetries.
These will be derived from the transformation properties of the zero curvature representation.

To start with, the constant astigmatism equation is invariant under the involution
$$
\numbered\label{invol}
\bar x = y, \quad \bar y = x, \quad \bar z = \frac1z,
$$
henceforth called {\it duality}.
By applying involution to the zero curvature representation $A\dif x + B\dif y$, we obtain 
$\bar A\dif\bar x + \bar B\dif\bar y = \bar B\dif x + \bar A\dif y$, where
$\bar A, \bar B$ result from $A,B$ by replacing $K_i,L_i$ with
$$ 
\numbered\label{bKL}
\begin{gathered}
\bar K_1 = \frac{z_x}{4 z^2}, &
\bar L_1 = \frac{z_y}4 + \frac{y}2, \\
\bar K_2 = \frac{y z_x}{2 z^2}, &
\bar L_2 = \frac{y z_y}2 - \frac{z}2 + \frac{y^2}{2}.
\end{gathered}
$$
Thus, the {\it dual conservation laws} will be derived by expansion of a vector
$\bar\Phi$ satisfying the system
$$
\numbered\label{bPhi sys}
\bar\Phi_x = \bar B \bar\Phi, \quad \bar\Phi_y = \bar A \bar\Phi.
$$
An easy computation reveals that the zero curvature representation
$\bar A(\lambda) \dif\bar x + \bar B(\lambda) \dif\bar y
 = \bar B(\lambda) \dif x + \bar A(\lambda) \dif y$ is gauge 
equivalent to
$A(\lambda) \dif x + B(\lambda) \dif y$.
The gauge matrix and its inverse are
$$
\numbered\label{H}
H = \frac1{\lambda - 2} \left(\begin{array}{cc} -x & 1 \\ 
-x y - 1 + \frac2\lambda & y
\end{array}\right),
\quad 
H^{-1} = 
\left(\begin{array}{cc} \lambda y & -\lambda 
\\ 
\lambda  x y + \lambda - 2 & -\lambda x 
\end{array}\right)
$$
unless $\lambda = 0,2$. 
Consequently, $H^{-1}\bar\Phi$ satisfies system~\eqref{Phi sys}.


Furthermore, the constant astigmatism equation is invariant under the reciprocal 
transformations~\cite{H-M III} 
$\mathcal X(x,y,z) = (x', y', z')$ and
$\mathcal Y(x,y,z) = (x^*, y^*, z^*)$, where
\begin{equation}
\label{RT}
\begin{gathered}
x' = \frac{x z}{x^2 z + 1}, &
y' = \eta, &
z' = \frac{(x^2 z + 1)^2}{z}, \\
x^* = \xi, & 
y^* = \frac{y}{y^2 + z}, &
z^* = \frac{z}{(y^2 + z)^2},
\end{gathered}
\end{equation}
see the Introduction for $\xi,\eta$. 
Neglecting the integration constants, we have 
$\mathcal X \circ \mathcal X = `Id = \mathcal Y \circ \mathcal Y$.
Since $\mathcal X$ is related to $\mathcal Y$ by the involution~\eqref{invol}, we shall
focus on $\mathcal X$ only.

The image of the zero curvature representation $A \dif x + B \dif y$ under $\mathcal X$ is 
$A' \dif x' + B' \dif y' = \tilde A \dif x + \tilde B \dif y$, where
$$ 
\label{tildeAB}
\tilde A(\lambda) = \left(\begin{array}{cc} 
-\frac{\lambda (\lambda - 2) K_1}{2 (\lambda - 1)} + \frac{\lambda^2 z L_1}{2 (\lambda - 1)} & 
\frac{\lambda (\lambda - 2) K_2}{2 (\lambda - 1)} - \frac{\lambda^2 z L_2}{2 (\lambda - 1)} 
\\ 
\frac{\lambda^2 z}{4(\lambda - 1)} &
\frac{\lambda (\lambda - 2) K_1}{2 (\lambda - 1)} - \frac{\lambda^2 z L_1}{2 (\lambda - 1)} 
\end{array}\right), 
\\
\tilde B(\lambda) = \left(\begin{array}{cc} 
-\frac{\lambda (\lambda - 2)L_1}{2 (\lambda - 1)} + \frac{\lambda^2 K_1}{2 (\lambda - 1) z} & 
\frac{\lambda (\lambda - 2) L_2}{2 (\lambda - 1)} - \frac{\lambda^2 K_2}{2 (\lambda - 1) z} 
\\
-\frac{\lambda (\lambda - 2)}{4 (\lambda - 1)} & 
\frac{\lambda (\lambda - 2)L_1}{2 (\lambda - 1)} - \frac{\lambda^2 K_1}{2 (\lambda - 1) z} 
\end{array}\right),
$$
with $K_i,L_i$ being given by formulas~\eqref{KL}.
Thus, the {\it reciprocal conservation laws} will be derived by expansion of a vector $\Phi'$ 
satisfying the system
$$
\numbered\label{Phi' sys}
\Phi'_x = \tilde A \Phi', \quad \Phi'_y = \tilde B \Phi'.
$$
It is easy to check that $\tilde A  \dif x + \tilde B \dif y$ is gauge equivalent to 
$A\dif x + B\dif y$ through the gauge matrix
\begin{equation}
\label{S}
S = \left(\begin{array}{cc} 0 & 1 \\ 1 & 0 \end{array}\right) = S^{-1}.
\end{equation}
Consequently, $S^{-1}\Phi' = S\Phi'$ satisfies system~\eqref{Phi sys}. 

Finally, {\it reciprocal dual conservation laws} are derived from
$\bar A' \dif \bar x' + \bar B' \dif \bar y'$. Omitting details, 
$P^{-1} \bar\Phi'$ satisfies system~\eqref{Phi sys}, where
$$
P^{-1} = \left(\begin{array}{cc} \frac{\lambda x z \eta}{x^2 z + 1} + \lambda - 2 & 
-\frac{\lambda x z}{x^2 z + 1} \\  \lambda \eta & -\lambda \end{array}\right).
$$

\section{The hierarchies}


Denote
$$
\Phi = (\begin{array}{c} u \\ v \end{array}), \quad
\bar\Phi = (\begin{array}{c} \bar u \\ \bar v \end{array}), \quad
\Phi' = (\begin{array}{c} u' \\ v' \end{array}), \quad
\bar\Phi' = (\begin{array}{c} \bar u' \\ \bar v' \end{array})
$$
the vectors generating ordinary, dual, 
reciprocal, and reciprocal dual hierarchy of nonlocal conservation laws.
Expanding $A(\lambda), B(\lambda)$ into series $\sum_i A_i \lambda^i$, $\sum_i B_i \lambda^i$ 
around $\lambda = 0$, we find that
$A_0 = B_0 = 0$, as expected, and 
$$
\begin{gathered}
A_1 = \left(\begin{array}{cc} 
K_1 & 0 \\ K_2 & -K_1 
\end{array}\right), &
A_2 = A_3 = \dots = \left(\begin{array}{cc} 
\frac12 K_1 + \frac12 z L_1 & - \frac14 {z}
\\ 
\frac12 K_2 + \frac12 z L_2 & - \frac12K_1 - \frac12 z L_1 
\end{array}\right),
\\
B_1 = \left(\begin{array}{cr} 
L_1 & -\frac12 \\ L_2 & -L_1 
\end{array}\right), &
B_2 = B_3 = \dots = \left(\begin{array}{cc} 
\frac12 L_1 + \frac12 K_1/z & -\frac14
\\ 
\frac12 L_2 + \frac12 K_2/z & -\frac12 L_1 - \frac12 K_1/z  
\end{array}\right).
\end{gathered}
$$
It follows that formulas~\eqref{Phi ser sys} simplify to
$$
\numbered\label{Cov Phi n}
\Phi_{n,x} = A_1 \Phi_{n-1} + A_2 \sum_{i = 0}^{n - 2} \Phi_{i}, 
\quad 
\Phi_{n,y} = B_1 \Phi_{n-1} + B_2 \sum_{i = 0}^{n - 2} \Phi_{i}. 
$$
Substituting 
$$
\Phi_0 = (\begin{array}{c} 1 \\ 0 \end{array}), \quad
\Phi_i = (\begin{array}{c} u_i \\ v_i \end{array}), \quad i > 1,
$$
into formulas~\eqref{Cov Phi n}, we immediately obtain the following construction of potentials
$u_i,v_i$.

\begin{construction} \label{con1}
Denote
$$ 
u_0 = 1, \qquad v_0 = 0, 
$$
and define potentials $u_{n},v_{n}$ by induction
$$
\numbered\label{Cov uv}
u_{n,x} = K_1 u_{n-1}
 + \frac12 (K_1 + z L_1) \sum_{i = 0}^{n - 2} u_i
 - \frac14 z \sum_{i = 0}^{n - 2} v_i,
\\
u_{n,y } = L_1 u_{n-1} - \frac12 v_{n-1}
 + \frac12 (L_1 + \frac{K_1}{z}) \sum_{i = 0}^{n - 2} u_i
 - \frac14 \sum_{i = 0}^{n - 2} v_i,
\\
v_{n,x} = K_2 u_{n-1} - K_1 v_{n-1}
 + \frac12 (K_2 + z L_2) \sum_{i = 0}^{n - 2} u_i
 - \frac12 (K_1 + z L_1) \sum_{i = 0}^{n - 2} v_i,
\\
v_{n,y} = L_2 u_{n-1} - L_1 v_{n-1}
 + \frac12 (L_2 + \frac{K_2}{z}) \sum_{i = 0}^{n - 2} u_i
 - \frac12 (L_1 + \frac{K_1}{z}) \sum_{i = 0}^{n - 2} v_i,
$$
for all $n > 0$, with $K_1,K_2,L_1,L_2$ being as introduced by formulas~\eqref{KL}
above.
\end{construction}

By construction, $u_{i},v_{i}$ are potentials of nonlocal conservation laws
of the constant astigmatism equation.
Observe that $u_{1,x} = K_1$, $u_{1,y} = L_1$, $v_{1,x} = K_2$, $v_{1,y} = L_2$ are
local.
As we shall see later, the potentials $u_{n},v_{n}$ are mutually independent.
Choosing a different initial vector $\Phi_0 \ne 0$, we obtain another set of 
potentials, linearly dependent on the potentials just constructed.

The constant astigmatism equation is invariant under the involution~\eqref{invol},
while Construction~\ref{con1} is not. 
Mutatis mutandis, we obtain the dual construction.

\begin{construction} \label{con2}
Denote
$$ 
\bar u_0 = 1, \qquad \bar v_0 = 0, 
$$
and define potentials $\bar u_{n},\bar v_{n}$ by induction
$$
\numbered\label{Cov tw}
\bar u_{n,x } = \bar L_1 \bar u_{n-1} - \frac12 \bar v_{n-1}
 + \frac12 (\bar L_1 + z \bar K_1) \sum_{i = 0}^{n - 2} \bar u_i
 - \frac14 \sum_{i = 0}^{n - 2} \bar v_i,
\\
\bar u_{n,y} = \bar K_1 \bar u_{n-1}
 + \frac12 (\bar K_1 + \frac{\bar L_1}{z}) \sum_{i = 0}^{n - 2} \bar u_i
 - \frac{1}{4z} \sum_{i = 0}^{n - 2} \bar v_i,
\\
\bar v_{n,x} = \bar L_2 \bar u_{n-1} - \bar L_1 \bar v_{n-1}
 + \frac12 (\bar L_2 + z \bar K_2) \sum_{i = 0}^{n - 2} \bar u_i
 - \frac12 (\bar L_1 + z \bar K_1) \sum_{i = 0}^{n - 2} \bar v_i,
\\
\bar v_{n,y} = \bar K_2 \bar u_{n-1} - \bar K_1 \bar v_{n-1}
 + \frac12 (\bar K_2 + \frac{\bar L_2}{z}) \sum_{i = 0}^{n - 2} \bar u_i
 - \frac12 (\bar K_1 + \frac{\bar L_1}{z}) \sum_{i = 0}^{n - 2} \bar v_i,
$$
for all $n > 0$, with $\bar K_1, \bar K_2, \bar L_1, \bar L_2$ being given by formulas
\eqref{bKL} above.
\end{construction}

By construction, $\bar u_{i}, \bar v_{i}$ are nonlocal potentials 
of the constant astigmatism equation. 
They are said to be {\it dual\/} to $u_i,v_i$.
However, $u_i,v_i,\bar u_i,\bar v_i$ are not functionally
independent, as we shall see below.

We shall not present any construction of potentials $u'_i,v'_i, \bar u'_i, \bar v'_i$.
Instead, we shall show how they depend on the potentials $u_i,v_i,\bar u_i,\bar v_i$
already constructed; see Proposition~\ref{prop:1} below.

\begin{remark} \rm \label{rem:lcl}
The potentials $\chi$, $\xi$, $\eta$, $\zeta$ (see the Introduction) are functions of
$x$, $y$, $z$, $u_1$, $v_1$, $w_1$, namely
$$
\begin{gathered}
\eta = -2 v_1, &
\chi = -4 u_1 + xy, \\
\xi = -2 w_1, &
\zeta = 2 u_1^2 + 2 u_1 - 4 u_2 - 2 y v_1 - \frac12 \ln z + \frac14 x^2 y^2 .
\end{gathered}
$$
\end{remark}

\section{Linear and Wronskian relations}

As we know from Section~\ref{sect:CAE}, vectors $\Phi$, $H^{-1}\Phi$, $S^{-1}\Phi$, $P^{-1}\Phi$ 
satisfy system~\eqref{Phi sys}. More generally, we can write
$$
\numbered\label{Phi k sys}
\Phi^{(k)}_x = A \Phi^{(k)}, \quad \Phi^{(k)}_y = B \Phi^{(k)},
\quad k = 1,\dots,m.
$$
Let $h$ be the dimension of the vectors $\Phi^{(k)}$, i.e., let $A,B$ belong to $\mathfrak{gl}(h)$. 
Let $W_{k_1 \dots k_h}$ denote the Wronskian determinant composed of 
$h$ columns $\Phi^{(k_1)}, \dots, \Phi^{(k_h)}$.
Denoting $c_{k_1 \dots k_h}$ arbitrary constants, possibly depending on $\lambda$, relations
$$
W_{k_1 \dots k_h} = c_{k_1 \dots k_h}
$$
are compatible with system~\eqref{Phi k sys}. 
Moreover, if $m > n$, then columns $\Phi^{(h+1)}, \dots, \Phi^{(m)}$
are linear combinations of $\Phi^{(1)}, \dots, \Phi^{(h)}$ with constant coefficients:
$$
\numbered\label{lin}
\Phi^{(l)} = c^l_{k_1} \Phi^{(k_1)} + \dots + c^l_{k_h} \Phi^{(k_h)}, \quad l = h + 1, \dots, m.
$$

\begin{example} \rm \label{ex}
Let $\mathfrak g = \mathfrak{sl}(2)$, let $\Phi^{(k)} = (\phi^{(k)}_1,\phi^{(k)}_2)^\top$ satisfy
\eqref{Phi k sys}.
For every pair of indices $1 \le k_1 < k_2 \le m$, let 
$W_{k_1 k_2} = \phi^{(k_1)}_1 \phi^{(k_2)}_2 - \phi^{(k_2)}_1 \phi^{(k_1)}_2$ denote
the corresponding Wronskian determinant. If $m = 2$, then we have one relation $W_{12} = c_{12}$, 
and we can solve it, say, for $\phi^{(2)}_2$ to get
$$
\numbered\label{phi22}
\phi^{(2)}_2 = \frac{\phi^{(1)}_2 \phi^{(2)}_1 + c_{12}}{\phi^{(1)}_1}.
$$
If $m \ge 3$, then, additionally to~\eqref{phi22}, we have linear relations
$$
c_{k_1k_2} \phi^{(k_3)}_1 - c_{k_1k_3} \phi^{(k_2)}_1 + c_{k_2k_3} \phi^{(k_1)}_1 = 0, \quad
c_{k_1k_2} \phi^{(k_3)}_2 - c_{k_1k_3} \phi^{(k_2)}_2 + c_{k_2k_3} \phi^{(k_1)}_2 = 0
$$ 
and we can express $\phi^{(k)}_i$, $k > 2$, in terms of $\phi^{(1)}_i, \phi^{(2)}_i$.
It is perhaps worth noticing that the Wronskian determinants $W_{k_1 k_2}$ 
themselves are not independent if $m \ge 4$.
As a consequence, relations 
$c_{k_1 k_2} c_{k_3 k_4} - c_{k_1 k_3} c_{k_2 k_4} + c_{k_1 k_4} c_{k_2 k_3} = 0$ are imposed,
one for each quadruple of mutually distinct indices $k_1, k_2, k_3, k_4$.
\end{example}

The linear relations are to be employed first.
We have
\begin{equation}
\begin{gathered}
\Phi^{(1)} = \Phi = \left(\begin{array}{cc} u \\ v \end{array}\right), 
&
\Phi^{(2)} = H^{-1} \bar\Phi = \left(\begin{array}{cc} \lambda  y u' - \lambda  v' 
\\ (\lambda  x y + \lambda - 2) u' - \lambda  x v' 
\end{array}\right), 
\\
\Phi^{(3)} = S \bar\Phi = \left(\begin{array}{cc} \bar v \\ \bar u \end{array}\right),
&
\Phi^{(4)} = P^{-1} \bar\Phi' = \left(\begin{array}{cc} 
(\frac{\lambda x z \eta}{x^2 z + 1} + \lambda - 2) \bar u' - \frac{\lambda x z}{x^2 z + 1} \bar v'
\\ \lambda  \eta \bar u' - \lambda \bar v' \end{array}\right).
\end{gathered}
\end{equation}
We obtain two linear relations~\eqref{lin}, namely
$$
\left(\begin{array}{cc} \bar v \\ \bar u \end{array}\right)
 = c^3_1 \left(\begin{array}{cc} u \\ v \end{array}\right) 
  + c^3_2 \left(\begin{array}{cc} \lambda  y u' - \lambda  v' \\
(\lambda  x y + \lambda - 2) u' - \lambda  x v' \end{array}\right),
\\
\left(\begin{array}{cc} 
(\frac{\lambda x z \eta}{x^2 z + 1} + \lambda - 2) \bar u' - \frac{\lambda x z}{x^2 z + 1} \bar v'
\\ \lambda  \eta \bar u' - \lambda \bar v' \end{array}\right)
 = c^4_1 \left(\begin{array}{cc} u \\ v \end{array}\right) 
  + c^4_2 \left(\begin{array}{cc} \lambda  y u' - \lambda  v' \\
(\lambda  x y + \lambda - 2) u' - \lambda  x v' \end{array}\right).
$$
The choice of constants $c^3_1,c^3_2,c^4_1,c^4_2$ is more or less arbitrary and influences 
the extension of $\mathcal X$ to higher potentials.
We would like to preserve the relation $\mathcal X \circ \mathcal X = `Id$.
In particular, we want $v_1' = -\frac12 y$, in view of $y' = \eta = -2 v_1$ from 
Remark~\ref{rem:lcl}.
For this reason, we have to withstand the temptation to set $v'_i = u_i$, $u'_i = v_i$ by 
choosing $c^3_1 = 1$, $c^3_2 = 0$.
We would also like $\mathcal X$ to preserve constants, i.e., we want
$u'_0 = 1$, $v'_0 = 0$, $\bar u'_0 = 1$, $\bar v'_0 = 0$. 
Led by these considerations, we set
$$
c^3_1 = 0, \quad 
c^3_2 = -\frac12, \quad
c^4_1 = -2, \quad
c^4_2 = 0.
$$ 
Under this choice, the above linear relations yield
\begin{equation}
\label{Xact}
\begin{gathered}
u' = (1 - \frac{\lambda}{2} (x y + 1) ) \bar u + \frac{\lambda }{2} x \bar v,
&
\bar u' = \frac{2}{2 - \lambda}(u - \frac{x z}{x^2 z + 1} v),
\\
v' = \-\frac{\lambda}{2}   y \bar u + \frac{\lambda}{2} \bar v,
&
\bar v' = \frac{2}{\lambda} v + \frac{4 v_1}{\lambda - 2} (u - \frac{x z}{x^2 z + 1} v).
\end{gathered}
\end{equation}
Moreover, $p'' = p$ for each potential $p = u, v, \bar u, \bar v$.

\begin{proposition}
\label{prop:1}
Potentials $u_i,v_i,\bar u_i,\bar v_i$ transform under $\mathcal X$ as follows:
$$
\begin{gathered}
u'_i = \bar u_i - \frac12 (1 + x y) \bar u_{i-1} + \frac12 x \bar v_{i-1}, 
\\
v'_i = -\frac12 y \bar u_{i-1} + \frac12 \bar v_{i-1}, 
\\
\bar u'_i = \sum_{j = 0}^{i} \frac1{2^{i-j}}(u_j - \frac{x z}{x^2 z+1} v_j), 
\\
\bar v'_i = 2 v_{i+1}
 - 2 v_1 \sum_{j = 0}^{i} \frac1{2^{i-j}}(u_j - \frac{x z}{x^2 z+1} v_j).
\end{gathered}
$$
Moreover, $p_i'' = p$ for each potential $p_i = u_i, v_i, \bar u_i, \bar v_i$.
\end{proposition}

\begin{proof}
By expanding formulas~\eqref{Xact} in powers of $\lambda$.
\end{proof}

\begin{proposition}
\label{prop:1}
Potentials $u_i,v_i,\bar u_i,\bar v_i$ transform under $\mathcal Y$ as follows:
$$
\begin{gathered}
u^*_i = \sum_{j = 0}^{i} \frac1{2^{i-j}}(\bar u_j - \frac{y}{y^2 + z} \bar v_j), 
\\
v^*_i = 2 v_{i+1}
 - 2 \bar v_1 \sum_{j = 0}^{i} \frac1{2^{i-j}}(\bar u_j - \frac{y}{y^2 + z} \bar v_j),
\\
\bar u^*_i = u_i - \frac12 (1 + x y) u_{i-1} + \frac12 y v_{i-1}, 
\\
\bar v^*_i = -\frac12 x u_{i-1} + \frac12 v_{i-1}. 
\end{gathered}
$$
Moreover, $p_i^{**} = p$ for each potential $p_i = u_i, v_i, \bar u_i, \bar v_i$.
\end{proposition}

\begin{proof}
By $\mathcal Y = \mathcal I \circ \mathcal X \circ \mathcal I$,
where $\mathcal I : x \otto y$, $z \otto 1/z$, $u_i \otto \bar u_i$, $v_i \otto \bar v_i$ is
the duality transformation.
\end{proof}

We see that equalities $\mathcal X \circ \mathcal X = `Id = \mathcal Y \circ \mathcal Y$ 
still hold after extension of $\mathcal X, \mathcal Y$ to the higher potentials.

Turning back to ordinary and dual conservation laws, we recall that
$H^{-1} \bar\Phi$ and $\Phi$ satisfy one and the same linear system~\eqref{Phi sys}. 
This implies constancy of the Wronskian determinant of $\Phi = (u,v)^\top$ and 
$H^{-1} \bar\Phi = H^{-1} (\bar u, \bar v)^\top$, which is equivalent to
$$
\numbered\label{rel}
R(u, v,\bar u, \bar v, \lambda)
 \equiv
   (\lambda - 2) u \bar u + \lambda (x u - v)(y \bar u - \bar v) = c(\lambda).
$$
It follows that hierarchies $u_i, v_i$, $\bar u_i, \bar v_i$ are
not independent.

\begin{proposition}
For all integers $n \ge -1$ and constants $c_i$, we have the relations
$$
\numbered\label{rels}
\sum_{i = 0}^n u_i \bar u_{n-i} 
   + \sum_{i = 0}^n (x u_i - v_i) (y \bar u_{n-i} - \bar v_{n-i})
  - 2 \sum_{i = 0}^{n+1} u_i \bar u_{n-i+1} 
  - c_{n+1}.
$$
\end{proposition}

\begin{proof}
The statement results from formula~\eqref{rel} by expanding around $\lambda = 0$ and setting
$c(\lambda) = \sum c_n \lambda^n$.
\end{proof}

We have the freedom to choose $c(\lambda) \ne 0$
(if $c = 0$ than the three hierarchies are functionally dependent again).
To have $\bar u_0 = 1$, we choose $c = -2$, i.e., 
$c_0 = -2$, $c_i = 0$ for $i \gt 0$. 
Solving \eqref{rels} with respect to $\bar u_{k + 1}$ and renaming 
$\bar v_i$ to $w_i$, we get the recursion formulas
$$ 
\numbered\label{bar u}
\bar u_0 = 1, \\
\bar u_{k + 1} =  
    \frac{1 + xy}2 \sum_{i = 0}^k u_i \bar u_{k-i}
    - \frac y2 \sum_{i = 0}^{k} \bar u_i v_{k-i} 
    +   \frac12 \sum_{i = 0}^{k} (v_i - x u_i) w_{k-i} 
    - \sum_{i = 1}^{k+1} u_i \bar u_{k-i+1}.
$$
E.g., $\bar u_1 = -u_1 + \frac12(1 + x y)$, etc.

Assuming assignments~\eqref{bar u} and renaming 
$\bar v_i$ to $w_i$,
systems~\eqref{Cov uv} and~\eqref{Cov tw} reduce to covering~\eqref{Cov uvw} below.

\begin{construction}
Let
$$ 
u_0 = 1, \qquad v_0 = w_0 = 0, 
$$
and define potentials $u_{n}$, $v_{n}$, $w_n$ by induction
$$
\numbered\label{Cov uvw}
u_{n,x} = K_1 u_{n-1}
 + \frac12 (K_1 + z L_1) \sum_{i = 0}^{n - 2} u_i
 - \frac14 z \sum_{i = 0}^{n - 2} v_i,
\\
u_{n,y } = L_1 u_{n-1} - \frac12 v_{n-1}
 + \frac12 (L_1 + \frac{K_1}{z}) \sum_{i = 0}^{n - 2} u_i
 - \frac14 \sum_{i = 0}^{n - 2} v_i,
\\
v_{n,x} = K_2 u_{n-1} - K_1 v_{n-1}
 + \frac12 (K_2 + z L_2) \sum_{i = 0}^{n - 2} u_i
 - \frac12 (K_1 + z L_1) \sum_{i = 0}^{n - 2} v_i,
\\
v_{n,y} = L_2 u_{n-1} - L_1 v_{n-1}
 + \frac12 (L_2 + \frac{K_2}{z}) \sum_{i = 0}^{n - 2} u_i
 - \frac12 (L_1 + \frac{K_1}{z}) \sum_{i = 0}^{n - 2} v_i, 
\\
w_{n,x} = \bar L_2 \bar u_{n-1} - \bar L_1 w_{n-1}
 + \frac12 (\bar L_2 + z \bar K_2) \sum_{i = 0}^{n - 2} \bar u_i
 - \frac12 (\bar L_1 + z \bar K_1) \sum_{i = 0}^{n - 2} w_i, 
\\ 
w_{n,y} = \bar K_2 \bar u_{n-1} - \bar K_1 w_{n-1}
 + \frac12 (\bar K_2 + \frac{\bar L_2}{z}) \sum_{i = 0}^{n - 2} \bar u_i
 - \frac12 (\bar K_1 + \frac{\bar L_1}{z}) \sum_{i = 0}^{n - 2} w_i,
$$
with $\bar u_i$ being given by formulas~\eqref{bar u}.
\end{construction}

By construction, equations~\eqref{Cov uvw} are compatible and yield a 
triple hierarchy of conservation laws of the constant astigmatism equation.


\begin{remark} \rm
The constant astigmatism equation possesses a scaling symmetry.
The variables have weights according to Table~1.
\begin{table}[h]
$$
\begin{array}{l|rrrrrr}
\text{Variable} & x & y & z & u_i & v_i & w_i \\\hline
\text{Weight}  & -1 & \hm1 & \hm2 &   \hm0 &  -1 &  \hm1
\end{array}
$$
\caption{Weights of variables under the scaling symmetry}
\end{table}
\end{remark}

\section{Independence}
\label{sect:indep}


The last part of this paper is devoted to the proof that there are no relations 
among $u_i, v_i,\bar u_i, \bar v_i$ other than~\eqref{rels}.
Actually we prove that there are no relations among $u_i, v_i, w_i$ determined by Construction~3.
Otherwise the method is that of recent works~\cite{K-S,K-S-M}.

Consider the base coordinates $x,y$ and the jet coordinates $z_\mu$.
The {\it total derivatives} are the vector fields
$$
\numbered\label{TD}
D_x = \frac{\partial}{\partial x}
 + \sum_\mu z_{\mu x} \frac{\partial}{\partial z_\mu}, \quad
D_y = \frac{\partial}{\partial y}
 + \sum_\mu z_{\mu y} \frac{\partial}{\partial z_\mu},
$$
where $\mu$ runs over all monomials in $x,y$ (according to a convenient, but abandoned tradition, 
$\mu$ can be thought of as a monomial $x^i y^j$).

To impose the constant astigmatism equation $\mathcal E$, we assign
$$
\numbered\label{assign}
z_{yy} = -(\frac1z)_{xx} - 2, \qquad
z_{\underbrace{x \dots x}_n y y} = -D_x^{n+2} (1/z), \quad n > 0.
$$
The first equation is the constant astigmatism equation solved for $z_{yy}$, the others 
are the differential consequences of the first equation. 
The total derivatives can be restricted to $\mathcal E$ (since they are tangent to $\mathcal E$).
Obviously, $D_x|_{\mathcal E},D_y|_{\mathcal E}$ commute. They are given by the same formulas~\eqref{TD}
under assignments~\eqref{assign} and with $z_\mu$ running only over the unassigned jet coordinates
(with $\mu$ running over all monomials not divisible by~$y^2$).

Joint systems~\eqref{Cov uv} and~\eqref{Cov tw} give a covering $\tilde{\mathcal E}$ 
equipped with total derivatives given by
$$
\tilde D^{(4)}_x = D_x 
 + \sum_{n = 1}^\infty (u_{n,x} \frac{\partial}{\partial u_{n}}
 + v_{n,x} \frac{\partial}{\partial v_{n}} 
 + \bar u_{n,x} \frac{\partial}{\partial \bar u_{n}} 
 + \bar v_{n,x} \frac{\partial}{\partial \bar v_{n}}), 
\\
\tilde D^{(4)}_y = D_y 
 + \sum_{n = 1}^\infty (u_{n,y} \frac{\partial}{\partial u_{n}}
 + v_{n,y} \frac{\partial}{\partial v_{n}}
 + \bar u_{n,y} \frac{\partial}{\partial \bar u_{n}} 
 + \bar v_{n,y} \frac{\partial}{\partial \bar v_{n}}).
$$
Reducing by relations~\eqref{rels}, we obtain covering~\eqref{Cov uvw} with total 
derivatives given by
$$
\tilde D^{(3)}_x = D_x 
 + \sum_{n = 1}^\infty (u_{n,x} \frac{\partial}{\partial u_{n}}
 + v_{n,x} \frac{\partial}{\partial v_{n}} 
 + w_{n,x} \frac{\partial}{\partial w_{n}}), 
\\
\tilde D^{(3)}_y = D_y 
 + \sum_{n = 1}^\infty (u_{n,y} \frac{\partial}{\partial u_{n}}
 + v_{n,y} \frac{\partial}{\partial v_{n}}
 + w_{n,y} \frac{\partial}{\partial w_{n}}).
$$
We are going to prove that there are no relations among $x,y,z_\mu,u_n,v_n,w_n$.
A covering $\tilde{\mathcal E}$ is said to be {\it differentially connected} in the sense 
of~\cite{Ig,Kra} if every function
$F : \tilde{\mathcal E} \to \mathbb R$ such that $\tilde D_x F = \tilde D_y F = 0$ is a
constant (possibly depending on the spectral parameter).

If $\tilde D_x F = \tilde D_y F = 0$ and $F$ is not a constant, then condition $F = c(\lambda)$ 
determines a relation compatible with the covering $\tilde{\mathcal E}$. 
Needless to say, the constant astigmatism equation itself is differentially connected.
Moreover, being differentially connected is a typical property of the
covering~\eqref{Phi 1 sys} associated with a non-degenerate zero curvature representation.

In what follows, $\delta_{ij}$ is the Kronecker delta.
The binomial coefficients
$$
\numbered\label{binom}
\binom{n}{k} = \prod_{i = 1}^k \frac{n - i + 1}{i}
$$
are defined for all integer values of $n,k$.
For nonnegative~$n$ they are nonzero if and only if $0 \le k \le n$.
For negative~$n$ they are zero if and only if $n < k < 0$.

\begin{proposition}
Let $F$ be a smooth function of a finite number of variables $x,y,z_\mu$, $u_n,v_n,w_n$, 
$n = 1,\dots,N$.
If $\tilde D_x^{(3)} F = \tilde D_y^{(3)} F = 0$, then $F$ is a constant.
\end{proposition}

\begin{proof}
Let $\tilde{\mathcal E}^{(N)}$ be the covering corresponding to the potentials 
$u_n,v_n,w_n$, $n = 1,\dots,N$, so that $\tilde{\mathcal E}^{(N)}$ is the domain of $F$.
Since $F$ depends on a finite number of variables, there is a maximal degree monomial $\nu = x^i$ 
such that $F$ depends on $z_{\nu x}$ or $z_{\nu y}$.  
Then we have
$$
\frac{\partial F}{\partial z_{\nu x}} = \frac\partial{\partial z_{\nu xx}} \tilde D_x F = 0, \quad
\frac{\partial F}{\partial z_{\nu y}} = \frac\partial{\partial z_{\nu xy}} \tilde D_y F = 0
$$
for all $i \ge 0$.
It follows that $F$ cannot actually depend on $z_{\nu x}$ or $z_{\nu y}$ for any $\nu$.
Therefore, $F$ can only depend on the $3 N + 3$ variables $x,y,z$ and $u_n,v_n,w_n$, where 
$n = 1,\dots, N$ for some $N \in \mathbb N$.
Since $u_{n,x},u_{n,y},v_{n,x},v_{n,y},w_{n,x},w_{n,y}$ are linear in $z_x,z_y$, so are 
$\tilde D_x F$ and $\tilde D_y F$,
and we can decompose
$$
\tilde D_x F = F_{11} z_x + F_{12} z_y + F_{10}, \quad
\tilde D_y F = F_{21} z_x + F_{22} z_y + F_{20},
$$
where $F_{ij}$ are independent of $z_x,z_y$.
It is easy to check that $F_{22} = F_{11}$ and $F_{12} = z^2 F_{21}$, while the others are independent.
Define vector fields $X_1,X_0,Y_1,Y_0$ by 
$$
X_1 F = F_{11}, \quad X_0 F = F_{10}, \quad 
Y_1 F = F_{21}, \quad Y_0 F = F_{20}.
$$

In coordinates,
$$
X_1 = 
\frac{\partial}{\partial z} 
 - \frac{1}{8z} \sum_{i = 0}^{N}  \sum_{j = 0}^{i - 2} u_j \frac{\partial}{\partial {u_i}}
 - \frac{1}{8z} \sum_{i = 0}^{N}  \sum_{j = 0}^{i - 2} (2 x u_j - v_j) \frac{\partial}{\partial {v_i}} \\ \qquad
 + \frac{1}{8z} \sum_{i = 0}^{N}  \sum_{j = 0}^{i - 2} (2 y \bar u_j - w_j) \frac{\partial}{\partial {w_i}}, \\
Y_1 = 
 -\frac{1}{8z^2} \sum_{i = 0}^{N}  \sum_{j = 0}^{i - 1} (1 + \delta_{j, i-1}) u_j \frac{\partial}{\partial {u_i}}
 - \frac{1}{8z^2} \sum_{i = 0}^{N}  \sum_{j = 0}^{i - 1} (1 + \delta_{j, i-1}) (2 x u_j - v_j) \frac{\partial}{\partial {v_i}} \\ \qquad
 + \frac{1}{8z^2} \sum_{i = 0}^{N}  \sum_{j = 0}^{i - 1} (1 + \delta_{j, i-1}) (2 y \bar u_j - w_j) \frac{\partial}{\partial {w_i}}, \\
X_0 = 
\frac{\partial}{\partial x} 
+ \frac{z}4 \sum_{i = 0}^{N} \sum_{j = 0}^{i-2} (x u_j - v_j) \frac{\partial}{\partial u_i}
+ \frac14 \sum_{i = 0}^{N} \sum_{j = 0}^{i-2} [(x^2 z+1) u_j - x z v_j] \frac{\partial}{\partial v_i} \\ \qquad
+ \frac14 \sum_{i = 0}^{N} 
    \sum_{j = 0}^{i-1} (1 + \delta_{j, i-1}) [(y^2 - z) \bar u_j  - y w_j] \frac{\partial}{\partial w_i}
, \\
Y_0 = \frac{\partial}{\partial y} 
 + \frac14 \sum_{i = 0}^{N} \sum_{j = 0}^{i - 1} (1 + \delta_{j, i-1}) (x u_j - v_j)
  \frac{\partial}{\partial u_i}
\\\qquad
 + \frac1{4z} \sum_{i = 0}^{N} \sum_{j = 0}^{i - 1} (1 + \delta_{j, i-1}) [(x^2 z - 1) u_j - x z v_j]
  \frac{\partial}{\partial v_i}
\\\qquad
 + \frac1{4z} \sum_{i = 0}^{N} \sum_{j = 0}^{i - 2} [(y^2 - z) \bar u_j - y w_j]
  \frac{\partial}{\partial w_i}.
$$

To prove that $F$ is a constant, it suffices to show that it is invariant under $3 N + 3$ 
linearly independent fields.
Since $F_{ij}$ are zero when $\tilde D_x F = \tilde D_y F = 0$, the function $F$ is invariant 
under the fields $X_1, X_0, Y_1, Y_0$. 
Then $F$ is also invariant under all nested Lie brackets of arbitrary depth.
Thus, to prove that $F$ is a constant, it suffices to construct $3 N + 3$ linearly independent 
nested Lie brackets.

To compute the first sequence $Z^\circ_{n}$ of nested brackets, we introduce the linear 
combinations
$$
Z_1 = z X_1 - z^2 Y_1 = 
z \frac{\partial}{\partial z} 
 + \frac14 \sum_{i = 1}^{N} u_{i - 1} 						\frac{\partial}{\partial {u_i}} 
 + \frac14 \sum_{i = 1}^{N} (2 x u_{i-1} - v_{i-1}) \frac{\partial}{\partial {v_i}} \\ \qquad
 - \frac14 \sum_{i = 1}^{N} (2 y \bar u_{i-1} - w_{i-1})	\frac{\partial}{\partial {w_i}}, 
\\
Z_0 = X_0 - z Y_0 = 
\frac{\partial}{\partial x}  
- z \frac{\partial}{\partial y}  
 - \frac{z}2 \sum_{i = 1}^{N} (x u_{i-1} - v_{i - 1}) \frac{\partial}{\partial {u_i}}
\\ \qquad
 + \frac12 \sum_{i = 1}^{N} [(1 - x^2 z) u_{i-1} + x z v_{i-1}] \frac{\partial}{\partial {v_i}} 
 + \frac12 \sum_{i = 1}^{N} [(y^2 - z)\bar u_{i-1} - y w_{i - 1}] \frac{\partial}{\partial {w_i}}
. 
$$ 
Starting from $Z^\circ_0 = Z_0 + 2 x z Z_1$, we compute recursively 
$$
Z^\circ_{n} = [Z^\circ_{n-1}, Z_1] + Z^\circ_{n-1} 
\\\quad = 
\frac{\partial}{\partial x}
+ \frac{z}{2^{n + 1}} \sum_{i = n+1}^{N} (v_{i-n-1} - x u_{i-n-1}) \frac{\partial}{\partial {u_i}} 
+ \frac{xz}{2^{n + 1}} \sum_{i = n+1}^{N}  (v_{i-n-1} - x u_{i-n-1}) \frac{\partial}{\partial {v_i}} 
\\\quad
- \sum_{i = 1}^{N} \sum_{j = 0}^{i-1} \frac1{(-2)^{i - j}} \binom{n + 1}{i - j} u_j \frac{\partial}{\partial {v_i}}  
- \sum_{i = 1}^{N} \sum_{j = 0}^{i-1} \frac1{(-2)^{i - j}} \binom{n}{ i - j - 1} (y^2 \bar u_j - y w_j) \frac{\partial}{\partial {w_i}}
\\\quad
+ \frac{z}{2^{n + 1}} \sum_{i = n}^{N} (2 \bar u_{i-n} - \bar u_{i-n-1}) \frac{\partial}{\partial {w_i}}
, \quad n > 0. 
$$
Observe that $n+1$ is the lowest $i$ to occur in ${\partial}/{\partial {u_i}}$.
To compute another sequence $X_n^+$ of nested brackets, we first introduce
$$
Z^+_{n} = \sum_{k = 0}^n (-1)^k \binom{n}{k} Z^\circ_{n+k} 
\\\quad = 
  \sum_{k = 0}^n \frac{(-1)^k z}{2^{n + k + 1}} \binom{n}{k}
    \sum_{i = n+k+1}^{N} (v_{i-n-k-1} - x u_{i-n-k-1}) \frac{\partial}{\partial {u_i}} 
\\\qquad
+ \sum_{k = 0}^n \frac{(-1)^k xz}{2^{n + k + 1}} \binom{n}{k}
    \sum_{i = n+k+1}^{N}  (v_{i-n-k-1} - x u_{i-n-k-1}) \frac{\partial}{\partial {v_i}} 
\\\qquad
+ \sum_{k = 0}^n \frac{(-1)^k z}{2^{n + k + 1}} \binom{n}{k}
    \sum_{i = n+k}^{N} (2 \bar u_{i-n-k} - \bar u_{i-n-k-1}) \frac{\partial}{\partial {w_i}}
\\\qquad
- \sum_{k = 0}^n \frac{(-1)^k}{2^{n + k}} \binom{n+1}{k}
    \sum_{i = n+k}^{N} u_{i-n-k} \frac{\partial}{\partial {v_i}}
\\\qquad
- \sum_{k = 0}^n \frac{(-1)^k}{2^{n + k}} \binom{n}{k-1}
    \sum_{i = n+k}^{N} (y^2 \bar u_{i-n-k} - y w_{i-n-k}) \frac{\partial}{\partial {v_i}}
, \quad n > 0 
$$
and then
$$
X_n^+ = [X_0, Z^+_{n}]
 = \sum_{k = 0}^N \frac{(-1)^{k+1}}{2^{n+k}} B(n-1,k-1)
     \sum_{i = n+k}^{N} u_{i-n-k} \frac{\partial}{\partial {u_{i}}}
\\\qquad
 + \sum_{k = 0}^N \frac{(-1)^{k+1}}{2^{n+k-1}} B(n-1,k-1)
     \sum_{i = n+k}^{N} (x u_{i-n-k} + v_{i-n-k}) \frac{\partial}{\partial {v_{i}}}
\\\qquad
 + \sum_{k = 0}^N \frac{(-1)^{k}}{2^{n+k-1}} B(n-1,k-1)
     \sum_{i = n+k}^{N} (y \bar u_{i-n-k} + w_{i-n-k}) \frac{\partial}{\partial {w_{i}}},
$$
where 
$$
\numbered\label{BB}
B(n,k)
 = \sum_{i = 0}^k (-1)^{i-k} \binom{n-k+i+1}{i}
 = \sum_{i = 0}^{\inline{\lfloor k/2 \rfloor}} \binom{n - 2 i}{k - 2 i}.
$$
Recall that binomial coefficients are defined by formula~\eqref{binom}.
For us, the important fact is that $B(n,k) = 0$ for $k < 0$ and arbitrary $n$.
For $k \le n$, the values $B(n,k)$ are nonzero since they constitute the last Mendelson 
triangle~\cite{Mend}, see sequence A035317 in the The On-Line Encyclopedia of Integer 
Sequences~\cite{OEIS}.
For $k > n + 1$, we have $B(n,k) = (-1)^i\,2^{i - n - 2}$, nonzero as well.
Finally, $B(n,n+1) = 1$ if $n \ge -1$ is odd and $B(n,n+1) = 0$ otherwise.

To finish the proof, we consider the  $3 N + 3$rd order determinant $\Delta$ 
composed of the coefficients of the fields
$$
Y_1, X_1^+, \dots,  X_{N - 1}^+, Z^\circ_{1}, \dots, Z^\circ_{2 N + 1}, Z_1, Y_0
$$
at $p$ with respect to the basis 
$$
\frac\partial{\partial u_1}, \dots, \frac\partial{\partial u_N}, 
\frac\partial{\partial w_1}, \dots, \frac\partial{\partial w_N}, 
-\frac\partial{\partial x},
\frac\partial{\partial v_1}, \dots, \frac\partial{\partial v_N},
\frac\partial{\partial y},
\frac\partial{\partial z}.
$$
Since the coefficients are rational functions in the variables $x,y,z, u_i,v_i,w_i$
(polynomial except in $z$), so is $\Delta$.
It follows that if $\Delta$ is nonzero at one point $p \in \tilde{\mathcal E}^{(N)}$, 
then it is nonzero in a dense subset of $\tilde{\mathcal E}^{(N)}$ 
(the zeroes of $\Delta$ constitute a proper algebraic 
submanifold in $\tilde{\mathcal E}^{(N)}$).

Recall that $u_0 = 1$, $v_0 = w_0 = 0$.
Let $p$ be the point
$x = y = u_1 = \dots = u_N = v_1 = \dots = v_N = w_1 = \dots = w_N = 0$, $z = 1$, then
$\bar u_i = 1/2^i$ for all $i \ge 0$.
It is easy to check that $\Delta(p)$ has enough zero entries to decompose into a product of
four simple determinants. 
As a matter of fact, 
a thorough inspection reveals that $\Delta(p)$ decomposes into the blocks 
$$
\Delta(p) = 
\left|
\begin{array}{ccccc}
M_{11} & 0 & 0 & M_{14} \\
0 & M_{22} & 0 & M_{24} \\
0 & M_{32} & M_{33} & M_{34} \\
0 & 0 & 0 & M_{44} 
\end{array}
\right|
$$
with a division after $N$th, $2N$th and $3N + 1$st row and column.
Therefore, $\Delta(p) = |M_{11}| \cdot |M_{22}| \cdot |M_{33}| \cdot |M_{44}|$.
The simplest $2 \times 2$ block $M_{44}$ corresponds to the fields $Z_1, Y_0$ and the basis
vectors $\partial/\partial y$, $\partial/\partial z$. 
It follows that $M_{44}$ is a unit matrix, hence $|M_{44}| = 1$.
The second simplest $N \times N$ block $M_{22}$ corresponds to the fields 
$Z^\circ_{1}, \dots, Z^\circ_{N}$ 
and the basis vectors $\partial/\partial w_1, \dots, \partial/\partial w_N$.
As such, $M_{22}$ is a diagonal matrix with the entries $2^{-1}, 2^{-2}, \dots ,2^{-N}$ and,
therefore, $|M_{22}| = 2^{-(N + 1)/2} \ne 0$.
Only slightly more involved is $M_{11}$, which corresponds to the fields 
$Y_1, X_1^+, \dots,  X_{N - 1}^+$
and the basis vectors $\partial/\partial u_1, \dots, \partial/\partial u_N$.
The first column $-\frac14, -\frac18, \dots, -\frac18$ corresponds to $Y_1$, whereas the rest 
corresponds to $X_n^+$ at $p$, which is
$$
\sum_{k = 0}^N \frac{(-1)^{k+1}}{2^{n+k}} B(n-1,k-1)
     \sum_{i = n+k}^{N} u_{i-n-k} \frac{\partial}{\partial {u_{i}}}.
$$
Here the only nonzero terms are those with $i - n - k = 0$ (since $u_i = 0$ except for $i = 0$). 
Hence, the element at the crossing 
of the $n + 1$st column and $i$th row is 
$$
\frac{(-1)^{k+1}}{2^{i}} B(n-1,i-n-1).
$$
By the description of $B$ following equation~\eqref{BB}, we see 
that $M_{11}$ is lower triangular with 
$-\frac14$, $2^{-2}, 2^{-3}, \dots ,2^{-N}$ on the diagonal and, again, $|M_{11}| \ne 0$.

Finally, the remaining block $M_{33}$ corresponds to the fields 
$Z^\circ_{N+1}, \dots, Z^\circ_{2N + 1}$ and the basis vectors 
$-\partial/\partial x, \partial/\partial v_1, \dots, \partial/\partial v_N$.
It follows that the $ij$th element of $M_{33}$ is 
$$
\frac{(-1)^i}{2^{i - 1}} \binom{N + j + 1}{i - 1}
$$
and, therefore, 
$$
|M_{33}| = \prod_{i = 1}^{N + 1} (-1)^i/2^{i - 1}
 \ne 0.
$$
Summing up, the  fields 
$Y_1$, $X_1^+, \dots,  X_{N - 1}^+$, $Z^\circ_{1}, \dots, Z^\circ_{2 N + 1}$, $Z_1$, $Y_0$ 
are linearly independent in a dense subset of the $3 N + 3$-manifold
$\tilde{\mathcal E}^{(N)}$ and, therefore, $F = `const$, which finishes the proof.
\end{proof}

\begin{corollary}
There is no possible functional dependence among the potentials $u_i,v_i,w_i$.
\end{corollary}

\begin{proof}
A relation $F(u_1,\dots,u_N,v_1,\dots,v_N,w_1,\dots,w_N) = 0$
would imply that $\tilde D^{(3)}_x F = \tilde D^{(3)}_y F = 0$.
\end{proof}

Cf.~\cite[Corollary~1]{K-S}.

\section{Conclusions}

We found four infinite series of linearly independent nonlocal conservation laws of the 
constant astigmatism equation. The hierarchies are closed with respect
to duality $\mathcal I$ and reciprocal transformations $\mathcal X, \mathcal Y$.
The corresponding potentials $u_i,v_i,\bar u_i, \bar v_i$ 
exhibit functional dependence of Wronskian type. 

\section*{Acknowledgements}

MM acknowledges the support from the GA\v{C}R under grant P201/12/G028,
AH the support from the Silesian university under SGS/2/2013.
Thanks are also due to I.S. Krasil'shchik and M. Pavlov for enlightening discussions.

\end{document}